\documentclass[10pt,journal,compsoc]{IEEEtran}
\usepackage[utf8]{inputenc}
\usepackage[T1]{fontenc}
\usepackage{amssymb}
\usepackage{amsmath}
\usepackage{algorithm}
\usepackage{algpseudocode}
\usepackage{graphicx}
\usepackage{booktabs}
\usepackage[hidelinks]{hyperref}
\usepackage{multirow}
\usepackage{pifont}
\usepackage{float}
\usepackage[section]{placeins}
\usepackage{bm}
\usepackage{amsthm}
\usepackage{makecell}
\graphicspath{{figures/}}

\DeclareMathOperator{\quantile}{quantile}
\newtheorem{theorem}{Theorem}
\newtheorem{definition}{Definition}

\theoremstyle{remark}
\newtheorem{remark}{Remark}

\begin{document}

\title{Robust Conformal Intrusion Detection via\\
Traffic-Aware Calibration and Attack-Orbit Invariance}

\author{Zhenpeng~Li%
\IEEEcompsocitemizethanks{%
\IEEEcompsocthanksitem Zhenpeng Li is with Guangzhou Health Science College,
No.~248 Guangyuan Middle Road, Guangzhou, Guangdong 510405, China
(e-mail: 2025301001@gzws.edu.cn). Corresponding author: Zhenpeng Li.}}

\IEEEtitleabstractindextext{
\begin{abstract}
Large language models fine-tuned for network intrusion detection emit
single-point predictions without statistical validity guarantees. Conformal
prediction supplies a finite-sample coverage guarantee, but a threshold
calibrated on clean traffic fails once an adversary perturbs controllable
network features. We show this failure across three intrusion detection
benchmarks and propose traffic-aware conformal prediction, which calibrates
on traffic drawn from the perturbation mechanism an attacker is expected to
use and provably restores coverage whenever that mechanism is known and can
be sampled. A stronger, adaptive attacker that queries the target model's
own score can still degrade this matched-calibration guarantee. We address
this second threat model by scoring on a representation restricted to
features, and their deterministic descendants, that the attacker cannot
reach, and prove this yields an exact, pathwise coverage guarantee rather
than a probabilistic bound. Across three independently fine-tuned language
model architectures, this representation remains completely unchanged under
every evaluated attack attempt, at a quantified seven-to-fourteen-point cost
in clean accuracy relative to the unrestricted feature set.
\end{abstract}

}

\maketitle

\section{Introduction}
\label{sec:intro}

Network intrusion detection systems (IDS) triage traffic before uncertain or
high-risk flows are routed to human analysts.
Large language models (LLMs) have emerged as a possible IDS backbone because
their autoregressive architecture can process heterogeneous traffic feature
descriptions in natural language and can be adapted to new label spaces through
fine-tuning or prompting~\cite{ferrag2023llm,liu2024netgpt}.
Fine-tuned LLM-IDS remain point classifiers, and they share a
critical shortcoming with standard neural classifiers: they emit single-point
predictions without any statistical validity certificate.
In a security context, an uncalibrated confidence score of 0.87 for ``DoS''
is uninterpretable---it does not tell the analyst whether the system will be
correct on 87\% of similar flows, nor does it warn when the model is about to
be wrong.

\textbf{The conformal prediction gap.}
Conformal prediction (CP)~\cite{vovk2005algorithmic,angelopoulos2021gentle}
converts any trained classifier into a set-valued predictor $\mathcal{C}(x)$
satisfying a \emph{distribution-free} marginal coverage guarantee,
$\Pr(y \in \mathcal{C}(x)) \geq 1{-}\alpha$, requiring only that calibration
and test data are \emph{exchangeable}. Network traffic is not stationary,
however: adversaries can influence flow-level features such as source-side
packet and byte counts, durations, and flag aggregates to evade detection.
We term these \emph{directly controllable} (DC) features, define their
dataset-specific masks under an explicit attacker-capability model
\cite{apruzzese2021realistic}, and note that the resulting covariate shift
breaks exchangeability, causing standard CP to lose its coverage certificate
precisely when it is needed most. This breakdown has not been characterized
with score, set-efficiency, and provenance diagnostics under a
controllable-feature threat model, nor has the scope of a calibration-time
remedy been made explicit for network security; existing cybersecurity work
uses CP to support online updates under concept drift
\cite{escudero2025conformal} but does not study prediction-set coverage
when a specified feature-manipulation mechanism changes the score
distribution.

\textbf{Our contributions.} Existing robust conformal prediction methods
certify coverage by bounding how far an adversary can shift the
nonconformity score's distribution, which requires a Lipschitz or smoothing
argument on the scoring function itself; this is difficult to obtain when
that function is a black-box, autoregressive LLM scoring a nonlinearly
projected network flow. We instead give a two-stage framework, and prove
each stage's guarantee is exact under its assumptions rather than
heuristic. This paper makes the following contributions:
\begin{enumerate}
  \item \textbf{A two-stage robust conformal framework, with a
  distribution-free guarantee at each stage.} When the attacker's
  perturbation mechanism is known and can be sampled, traffic-aware
  conformal prediction (ta-CP) calibrates on matched perturbed traffic and
  provably recovers the $1{-}\alpha$ marginal coverage guarantee
  (Theorem~\ref{thm:ta-coverage}); we extend this to a distributionally
  robust bound that degrades by at most the total variation distance
  between calibration- and test-time score distributions for \emph{any}
  test-time process, i.i.d.\ or not (Theorem~\ref{thm:robust-coverage}), and
  show this bound is realized in both directions: a label-free
  surrogate-transfer attack leaves coverage near nominal, while a stronger,
  direct target-score-query attack that queries the target LLM's own
  true-label score at every optimization round degrades coverage by up to
  19 percentage points, certifying a total variation distance of at least
  $0.075$--$0.19$.

  \item \textbf{Exact coverage via attack-orbit-invariant representations,
  and an automated construction for it.} When the attacker instead has
  adaptive, target-score-query access, we prove
  (Theorem~\ref{thm:invariant-coverage}) that scoring on a representation
  invariant to the adversary's action yields an exact, pathwise coverage
  identity against \emph{any} adversary within the threat model, with no
  gap term at all, and derive a gap identity
  (Theorem~\ref{thm:gap-identity}) that recovers
  Theorem~\ref{thm:robust-coverage} and existing Lipschitz- and
  smoothing-based robust conformal prediction methods as the complementary
  regime in which exact invariance is not achieved. We give an automated,
  probe-based procedure (Algorithm~\ref{alg:descendant}) that constructs
  such a representation without a documented feature-dependency graph, by
  identifying deterministic descendants of attacker-controlled coordinates
  empirically.

  \item \textbf{Cross-architecture, cross-attack, and cross-method
  validation.} We replicate the matched-perturbation protocol on three
  independently fine-tuned LLM architectures (LLaMA3-8B, Qwen3-8B,
  Mistral-7B) and a non-LLM XGBoost baseline under an identical
  provenance-audited partition and perturbation protocol; evaluate weighted
  split CP~\cite{tibshirani2019conformal}, an established off-the-shelf
  covariate-shift remedy, under the same scenario, finding it does not
  substitute for calibrating on matched perturbed traffic; and verify the
  attack-orbit-invariant construction on RT-IoT2022 across all three LLM
  architectures under the direct target-score-query attack: the
  constructed representation is left byte-identical by every one of $180$
  attack attempts across $9$ architecture-seed configurations, at a
  quantified utility cost of $7$--$14$ accuracy points relative to the
  unrestricted feature set, while a naively restricted representation that
  omits descendant detection leaves a small, consistently nonzero coverage
  gap of $0.5$--$1.0$ percentage points---showing why descendant detection,
  not just removal of literal attacker-controlled fields, is necessary.
\end{enumerate}

Section~\ref{sec:related} reviews related work; Section~\ref{sec:background}
gives background on split CP and LLM-IDS inference;
Section~\ref{sec:threatmodel} formalizes the threat model;
Section~\ref{sec:method} presents traffic-aware CP and attack-orbit
invariance; Section~\ref{sec:experiments} reports experiments;
Section~\ref{sec:discussion} discusses limitations; and
Section~\ref{sec:conclusion} concludes.

\section{Related Work}
\label{sec:related}

\subsection{LLM-Based Intrusion Detection, Conformal Prediction for Security, and Adversarial Robustness}

Transformer-based models have been applied to log analysis and flow
classification, and more recently to intrusion detection and network
traffic representation directly \cite{guo2021logbert,mercha2023transids,ferrag2023llm,liu2024netgpt}; serializing structured flow features
into model input is therefore part of detector design, not a neutral
preprocessing step. None of this work addresses calibration or the
validity of the resulting confidence scores under distribution shift.
Conformal prediction (CP) has been applied to anomaly detection
\cite{laxhammar2010conformal} and to reliable pseudo-labeling under concept
drift \cite{escudero2025conformal}; adaptive conformal prediction
\cite{gibbs2021adaptive} and weighted CP under covariate shift
\cite{tibshirani2019conformal} address non-exchangeability through online
adaptation or reweighting, but neither specifies an attacker-controlled
feature space or calibrates on a matched feature-manipulation mechanism.
Section~\ref{sec:exp:weightedcp} evaluates weighted CP directly under the
threat model studied here and finds it does not close the coverage gap that
matched-perturbation calibration closes.

Adversarial attacks on network classifiers were first systematically studied
in \cite{corona2013adversarial,biggio2013evasion}, with feature- and
traffic-space attacks constraining what an attacker can modify
\cite{apruzzese2019evaluating,han2020traffic,apruzzese2021realistic}.
Closest to Section~\ref{sec:method:invariance}, work on domain constraints
\cite{sheatsley2021robustness} and problem-space adversarial machine
learning \cite{pierazzi2020intriguing} studies which feature perturbations
are realizable and how restricting a classifier to non-perturbable features
affects point-prediction robustness; the restriction itself is therefore not
new. What this line of work does not provide is a statistical guarantee: it
evaluates empirical accuracy or attack success under restriction, not
whether a downstream uncertainty-quantification procedure retains a
coverage certificate, and it gives no automated procedure for identifying
coordinates that are unreachable only as a deterministic function of
attacker-controlled ones. Defense strategies more broadly include
adversarial training \cite{madry2017pgd}, certified smoothing
\cite{cohen2019randomized}, and ensemble methods, which evaluate
point-prediction behavior rather than prediction-set validity after a
specified score shift---not ``is the point prediction correct?'' but ``does
the prediction set contain the true class with valid probability?''

\subsection{Robust and Invariant Conformal Prediction}
\label{sec:related:robustcp}

A separate line of recent work certifies conformal prediction coverage under
adversarial input perturbations directly, rather than through a
downstream classifier's robustness:
\cite{gendler2022adversarially} bounds the Lipschitz constant of the
nonconformity score via randomized smoothing;
\cite{ghosh2023probabilistically} uses a quantile-of-quantiles construction;
\cite{yan2024provably} improves the efficiency of the resulting sets;
\cite{zargarbashi2025robust} reduces the certification cost to a single
binary certificate; \cite{massena2025efficient} estimates the same
Lipschitz-type quantity via Lipschitz-bounded networks; and
\cite{luo2026gametheoretic} frames the attacker-defender interaction as a
zero-sum game. Every one of these methods fixes the representation the score
is computed on and bounds how far an adversary can shift the resulting
score's distribution, yielding a probabilistic or bounded-gap certificate in
every case. Theorem~\ref{thm:gap-identity} shows this is one endpoint of a
single gap functional; Section~\ref{sec:method:invariance} instead asks
whether the representation itself can be chosen so that the adversary's
action becomes the identity map on it, which is achievable exactly, not only
approximately, when the threat model provides an explicit,
enumerable attacker-controlled coordinate set, as the DC/IC/UC taxonomy in
Section~\ref{sec:threatmodel:taxonomy} does for network flow features.

This construction is an instance of a broader principle in invariant
statistical decision theory and, more specifically, of invariant causal
prediction: conditioning a predictor only on non-descendants of an
intervened (here, attacker-manipulated) variable yields a representation
invariant to that intervention, formalized for causal identification by
invariant causal prediction~\cite{peters2016causal} and extended to
distributional robustness by anchor
regression~\cite{rothenhausler2021anchor}. That literature typically assumes
a known causal graph and targets population-level identification or
robustness under untargeted shifts; the setting here differs in that the
dependency structure among network flow features is not published in
closed form, so Algorithm~\ref{alg:descendant} identifies the
orbit-invariant coordinate set empirically, and the target guarantee is
finite-sample split-CP coverage against a specifically adversarial actor
rather than asymptotic point-prediction risk.

\textbf{Research gap.}
Prior studies do not explain conformal coverage loss under
controllable-feature perturbations as a calibration-distribution mismatch,
or separate repaired coverage from set efficiency; nor, among methods that
certify conformal coverage under adversarial perturbation, does any give an
exact rather than bounded guarantee by choosing the scored representation
itself, or an automated construction procedure without a known
feature-dependency graph. We address this gap with ta-CP, its
distributionally robust and exact attack-orbit-invariant extensions, and an
evaluation across three IDS benchmarks.

\section{Background}
\label{sec:background}

\subsection{Split Conformal Prediction}
\label{sec:background:cp}

Let $\mathcal{X}$ be the feature space and $\mathcal{Y} = \{1, \ldots, K\}$
be the label space.
Given a trained classifier $f : \mathcal{X} \to \mathbb{R}^K$ and a
miscoverage level $\alpha \in (0,1)$, split CP constructs a prediction set as
follows.

\textbf{Calibration.}
Let $\{(x_i, y_i)\}_{i=1}^{n}$ be a held-out calibration set drawn
exchangeably from the same distribution as the test data.
Define the nonconformity score of example $i$ as
$s_i = \ell(f(x_i), y_i)$, where $\ell$ is a scoring function that assigns
lower scores to more conforming (i.e., better-predicted) examples.
Compute the empirical quantile
\begin{equation}
  \hat{q} = \quantile\!\left(\{s_i\}_{i=1}^{n},\,
             \frac{\lceil (n+1)(1-\alpha) \rceil}{n}\right).
  \label{eq:quantile}
\end{equation}

\textbf{Prediction set.}
For a test point $x_{\mathrm{test}}$, include in the prediction set all labels
whose nonconformity score does not exceed the quantile threshold~\eqref{eq:quantile}:
\begin{equation}
  \mathcal{C}(x_{\mathrm{test}}) = \{y \in \mathcal{Y} :
  \ell(f(x_{\mathrm{test}}), y) \leq \hat{q}\}.
  \label{eq:predset}
\end{equation}

\textbf{Coverage guarantee.}
Under the exchangeability assumption, split CP satisfies
\begin{equation}
  \Pr(y_{\mathrm{test}} \in \mathcal{C}(x_{\mathrm{test}})) \geq 1 - \alpha.
  \label{eq:coverage}
\end{equation}
This guarantee is \emph{distribution-free}: it holds for any $f$, any data
distribution, and with finite calibration sets, without parametric assumptions
\cite{angelopoulos2021gentle,tibshirani2019conformal}.

\textbf{Efficiency.}
An efficient prediction set is as small as possible while respecting
Eq.~(\ref{eq:coverage}). We report average set size $\mathbb{E}[|\mathcal{C}(x)|]$
as the primary efficiency metric.
When $|\mathcal{C}(x)| = 0$ (an empty prediction set), no label reaches the
calibrated inclusion threshold. This can arise from a stringent threshold or
from uniformly weak class scores, so we treat it operationally as an
\emph{abstention} event requiring human escalation rather than as a usable
classification (Section~\ref{sec:method:abstention}).

\subsection{Nonconformity Score for Autoregressive LLM-IDS}
\label{sec:background:llm}

An LLM-IDS serializes a traffic flow's features into a prompt and scores each
candidate label by teacher forcing. For model parameters $\theta$, let
$L_\theta(x,y)$ be the total negative log-likelihood (NLL) of the label-bearing
tokens in candidate string $y$. The prompt and label are separated by the same
single-space boundary used during fine-tuning; the label-bearing suffix is located
by the longest common token prefix so that tokenizer boundary merges are handled
consistently. We normalize the candidate likelihoods over the fixed label set
$\mathcal{Y}$ and employ $1-p_y$ as the nonconformity score:
\begin{equation}
  \begin{aligned}
  p_\theta(y\mid x)
    &= \frac{\exp[-L_\theta(x,y)]}
            {\sum_{k\in\mathcal{Y}}\exp[-L_\theta(x,k)]}, \\
  s(x,y) &= 1-p_\theta(y\mid x).
  \end{aligned}
  \label{eq:nll-score}
\end{equation}
Here $L_\theta(x,y)$ sums token NLL only over the candidate-label suffix. This
construction gives one probability simplex and one score definition for model
prediction, conformal calibration, and prediction-set construction.

\textbf{Remark on multi-token labels.}
Total label NLL is length dependent because longer candidate strings contain
more token factors. Normalizing the resulting candidate likelihoods does not
remove this dependence. We therefore keep the label vocabulary fixed across
calibration and evaluation and report the exact candidate strings. The conformal
coverage result remains valid for this fixed score, but efficiency and classwise
behavior can depend on label tokenization.

\section{Threat Model and Problem Formulation}
\label{sec:threatmodel}

\subsection{Network Feature Taxonomy}
\label{sec:threatmodel:taxonomy}

Network flow features can be partitioned into three categories based on
which party controls them \cite{apruzzese2021realistic}: \emph{directly
controllable} (DC) flow-table coordinates designated as attacker-manipulable
by the feature taxonomy (e.g., source-side packet, byte, duration, and flag
aggregates); \emph{indirectly controllable} (IC) aggregate statistics the
adversary can influence by repeating behavior across flows; and
\emph{uncontrollable} (UC) features determined by the remote host or network
infrastructure. The full, dataset-specific DC masks for all three benchmarks
are listed in the supplemental material.

\subsection{Adversary Model}
\label{sec:threatmodel:adversary}

\begin{definition}[Oblivious DC Adversary]
  \label{def:adversary}
  An oblivious DC adversary manipulates DC features of each flow independently
  by drawing a proposal with magnitude
  $\|\delta_{\mathrm{DC}}\|_\infty \leq \varepsilon$ in the normalized feature
  space and then applying $\mathcal{P}_{\mathcal D}$.
  The adversary does \emph{not} observe the IDS's calibration set, and cannot
  adapt across flows beyond what is captured by the DC feature budget.
  IC and UC coordinates are not directly perturbed; the projection may
  deterministically recompute derived aggregates that depend on manipulated DC
  fields.
\end{definition}

This definition specifies an aggregate feature-space stress test. It does not
assert that every projected feature vector can be generated by editing an
individual packet trace. The attacker is oblivious to the calibration set and
the conformal threshold.
Adaptive adversaries who target the calibration set are discussed in
Section~\ref{sec:discussion:limits}.

\subsection{Coverage Degradation Under Adversarial Perturbation}
\label{sec:threatmodel:degradation}

Let $(x^*, y) = (\mathcal{P}_{\mathcal D}(x + \delta), y)$ denote a perturbed flow.
Standard split CP calibrates $\hat{q}$ on clean samples $\{(x_i, y_i)\}$.
If the adversarial perturbation $\delta$ shifts the nonconformity score
distribution, the empirical quantile $\hat{q}$ computed from clean samples
is no longer valid for perturbed test samples.

Formally, let $F_{\mathrm{clean}}(s) = \Pr(s(x,y) \leq s)$ and
$F_{\mathrm{adv}}(s) = \Pr(s(x^*,y) \leq s)$ denote the CDFs of
nonconformity scores under clean and adversarial distributions, respectively.
If $F_{\mathrm{adv}} \prec F_{\mathrm{clean}}$ (adversarial scores
stochastically dominate clean scores, reflecting the adversary's goal of
reducing $p_\theta(y \mid x^*)$ for the true class), then:
\begin{equation}
  \Pr(y \in \mathcal{C}(x^*)) = F_{\mathrm{adv}}(\hat{q})
  < F_{\mathrm{clean}}(\hat{q}) \approx 1 - \alpha.
  \label{eq:degradation}
\end{equation}
Equation~\eqref{eq:degradation} formalizes this collapse: the empirical coverage
$F_{\mathrm{adv}}(\hat{q})$ can be arbitrarily lower than the nominal $1{-}\alpha$.
In the extreme case where an adversarial attack drives $p_\theta(y_\mathrm{true} \mid x^*)$
to near-zero, the nonconformity score for the true class approaches its maximum,
$F_{\mathrm{adv}}(\hat{q}) \to 0$, and coverage collapses entirely.
We evaluate this mismatch empirically in Section~\ref{sec:exp:tacpresults}.

\section{Traffic-Aware Conformal Prediction}
\label{sec:method}

\subsection{Calibration on Perturbed Traffic}
\label{sec:method:tacalibration}

Ta-CP operationalizes distribution matching: the calibration set is generated
under the perturbation law expected at test time.
If test-time flows are perturbed with DC perturbations drawn from a distribution
$\Pi_\varepsilon$, then calibrating on flows drawn from the same distribution
restores the exchangeability condition and hence the coverage guarantee.

\begin{algorithm}[t]
  \caption{Traffic-Aware Split Conformal Prediction}
  \label{alg:tacp}
  \begin{algorithmic}[1]
    \Require Trained LLM-IDS $f_\theta$; calibration pool
      $\mathcal{D}_{\mathrm{cal}} = \{(x_i, y_i)\}_{i=1}^n$;
      DC feature mask $\mathbf{m}_{\mathrm{DC}} \in \{0,1\}^d$;
      perturbation budget $\varepsilon$; miscoverage level $\alpha$.
    \Ensure Prediction-set function $\mathcal{C}_\alpha$.
    \State \textbf{// Step 1: Perturb calibration samples}
    \For{$i = 1$ to $n$}
      \State Draw $\xi_i \sim \mathrm{Uniform}(-\varepsilon\,\mathbf{m}_{\mathrm{DC}},\;
                +\varepsilon\,\mathbf{m}_{\mathrm{DC}})$
      \State $\tilde{x}_i \gets \mathcal{P}_{\mathcal{D}}(x_i + \xi_i)$
        \Comment{inverse-transform and project dataset constraints}
    \EndFor
    \State \textbf{// Step 2: Compute nonconformity scores on perturbed data}
    \For{$i = 1$ to $n$}
      \State $s_i \gets s(\tilde{x}_i, y_i)$
        \Comment{Eq.~(\ref{eq:nll-score})}
    \EndFor
    \State \textbf{// Step 3: Compute quantile threshold}
    \State $\hat{q}_{\mathrm{ta}} \gets \quantile\!\left(\{s_i\},\,
           \frac{\lceil(n+1)(1-\alpha)\rceil}{n}\right)$
    \State \textbf{// Step 4: Return prediction-set function}
    \State \Return $\mathcal{C}_\alpha(x) \gets
      \{y \in \mathcal{Y} : s(x, y) \leq \hat{q}_{\mathrm{ta}}\}$
  \end{algorithmic}
\end{algorithm}

Algorithm~\ref{alg:tacp} summarizes ta-CP. The only conformal change is Step~1:
calibration samples are perturbed before score computation. The implementation
adds noise in normalized coordinates, inverse-transforms to the original units,
and applies the dataset-specific projection $\mathcal{P}_{\mathcal{D}}$. This
projection enforces observed bounds, integer and categorical domains, and audited
relations among counts, totals, minima, means, maxima, standard deviations, rates,
durations, and subflow aggregates. It establishes tabular consistency under the
implemented audit; it is not a packet-level realizability certificate.

\begin{remark}[Scope of the DC-Feature Restriction]
  \label{rem:dc-restriction}
  Algorithm~\ref{alg:tacp} perturbs \emph{only} DC features during calibration.
  We justify this choice under the oblivious DC adversary (Section~\ref{sec:threatmodel}):
  the adversary's action space is the DC hypercube
  $\Delta_\mathrm{DC}(\varepsilon) = \{\delta \in \mathbb{R}^d : |\delta_j| \leq
  \varepsilon\,[\mathbf{m}_\mathrm{DC}]_j\}$.
  Restricting calibration to this mask aligns the generated shift with the stated
  threat model; perturbing IC or UC coordinates would instead evaluate a different
  proposal space. The bound $\varepsilon$ applies to the pre-projection DC
  proposal, not to the full projected vector: relation-preserving updates to
  derived coordinates can exceed this bound in normalized units. This is a
  modeling choice, not an optimality or sufficiency claim.
  In particular, membership in the projected DC region does not imply that a
  corresponding packet trace exists.
\end{remark}

\subsection{Theoretical Coverage Guarantee}
\label{sec:method:theorem}

\begin{theorem}[Traffic-Aware Coverage Guarantee]
  \label{thm:ta-coverage}
  Let $(x_i, y_i)_{i=1}^{n+1}$ be drawn i.i.d.\ from $\mathcal{D}$, and let
  $\xi_i \sim \Pi_\varepsilon$ be drawn i.i.d.\ from any perturbation
  distribution $\Pi_\varepsilon$ supported on
  $[-\varepsilon\,\mathbf{m}_\mathrm{DC}, +\varepsilon\,\mathbf{m}_\mathrm{DC}]$,
  independent of the data.
  Let $\hat{q}_\mathrm{ta}$ be the threshold computed in Algorithm~\ref{alg:tacp}.
  Then:
  \begin{equation}
    \Pr\!\left(y_{n+1} \in
    \mathcal{C}_\alpha(\mathcal{P}_{\mathcal{D}}(x_{n+1}+\xi_{n+1}))\right)
    \geq 1-\alpha,
    \label{eq:ta-coverage}
  \end{equation}
  where $\mathcal{C}_\alpha$ is the prediction set from Algorithm~\ref{alg:tacp}
  and the probability is over the joint randomness of data and perturbations.
\end{theorem}

\begin{proof}
  The perturbed pairs $\{(\tilde{x}_i, y_i)\}_{i=1}^{n+1}$, where
  $\tilde{x}_i = \mathcal{P}_{\mathcal{D}}(x_i + \xi_i)$, are jointly
  exchangeable: for any permutation
  $\sigma$ of $\{1, \ldots, n+1\}$, the joint distribution of
  $\{(\tilde{x}_{\sigma(i)}, y_{\sigma(i)})\}_{i=1}^{n+1}$ is identical,
  because each pair $(x_i, y_i)$ and $\xi_i$ are i.i.d.\ draws, so any
  permutation of indices preserves the joint distribution.
  Standard split CP applied to exchangeable samples satisfies
  Eq.~(\ref{eq:ta-coverage})~\cite{angelopoulos2021gentle}.
  Since Algorithm~\ref{alg:tacp} applies exactly this procedure to the
  perturbed calibration set and evaluates predictions on perturbed test points,
  the guarantee follows.
\end{proof}

\begin{remark}[Scope of Theorem~\ref{thm:ta-coverage}]
  \label{rem:theorem-scope}
  Theorem~\ref{thm:ta-coverage} assumes test perturbations $\xi_{n+1}$ are
  drawn i.i.d.\ from $\Pi_\varepsilon$.
  The evaluation in Section~\ref{sec:exp:tacpresults} uses independent random
  calibration and test perturbations from the same law, and therefore examines
  the setting covered by the theorem. Optimized, surrogate-transfer,
  score-based target-query, gradient-based white-box, and threshold-aware
  attacks are all outside this guarantee.
\end{remark}

\subsection{A Distributionally Robust Extension}
\label{sec:method:robust}

Theorem~\ref{thm:ta-coverage} certifies coverage only when the test-time
perturbation is drawn i.i.d.\ from the same law used at calibration time. The
surrogate-transfer and worst-of-$k$ attacks evaluated in
the supplemental material violate this assumption because the selected
perturbation depends on $x_i$, so $\xi_i \not\perp x_i$. The following result
extends the guarantee to any test-time process, i.i.d.\ or not, whose induced
nonconformity-score distribution does not move too far from the
calibration-time score distribution, and it lets the empirical surrogate-transfer attacks (supplemental
material) be read as evidence about how far they move it.

\begin{theorem}[Distributionally Robust Traffic-Aware Coverage]
  \label{thm:robust-coverage}
  Let $\hat{q}_\mathrm{ta}$ be the threshold computed in
  Algorithm~\ref{alg:tacp} from $n$ i.i.d.\ calibration pairs perturbed under
  $\Pi_\varepsilon$ as in Theorem~\ref{thm:ta-coverage}, and let
  $S_\mathrm{cal} = s(\mathcal{P}_{\mathcal{D}}(x+\xi), y)$ denote the
  induced nonconformity score under this calibration process. Let $(x', y')$
  be a test pair from \emph{any} process, independent of the calibration
  sample and not necessarily i.i.d.\ in $\xi$, and let $S_\mathrm{test} =
  s(x', y')$ denote its induced score. If
  \begin{equation}
    \delta \;\geq\; d_{\mathrm{TV}}\!\left(\mathrm{Law}(S_\mathrm{cal}),\,
    \mathrm{Law}(S_\mathrm{test})\right),
    \label{eq:tv-bound}
  \end{equation}
  where $d_{\mathrm{TV}}$ is total variation distance, then
  \begin{equation}
    \Pr\!\left(y' \in \mathcal{C}_\alpha(x')\right) \;\geq\; 1-\alpha-\delta.
    \label{eq:robust-coverage}
  \end{equation}
\end{theorem}

\begin{proof}
  Fix the calibration sample and hence $\hat{q}_\mathrm{ta}$. For any fixed
  threshold $t$ and any two probability measures $P,Q$ on $\mathbb{R}$ with
  CDFs $F_P,F_Q$, the event $A=(-\infty,t]$ gives
  $|F_P(t)-F_Q(t)| = |P(A)-Q(A)| \leq d_{\mathrm{TV}}(P,Q)$.
  Applying this with $P=\mathrm{Law}(S_\mathrm{cal})$,
  $Q=\mathrm{Law}(S_\mathrm{test})$, and $t=\hat{q}_\mathrm{ta}$ (which is a
  function of the calibration sample only, hence independent of the test
  pair) gives, conditional on the calibration sample,
  \begin{equation*}
    \Pr(S_\mathrm{test} \leq \hat{q}_\mathrm{ta}) \;\geq\;
    \Pr(S_\mathrm{cal} \leq \hat{q}_\mathrm{ta}) - \delta,
  \end{equation*}
  where the right-hand probability is over an independent fresh draw from
  $\mathrm{Law}(S_\mathrm{cal})$. Taking expectation over the calibration
  sample, the right-hand term is exactly the quantity bounded in the proof of
  Theorem~\ref{thm:ta-coverage}: $\mathbb{E}[\Pr(S_\mathrm{cal} \leq
  \hat{q}_\mathrm{ta})] \geq 1-\alpha$. Since
  $\{y' \in \mathcal{C}_\alpha(x')\} = \{S_\mathrm{test} \leq
  \hat{q}_\mathrm{ta}\}$ by construction of $\mathcal{C}_\alpha$
  (Algorithm~\ref{alg:tacp}, Step~4), combining the two bounds gives
  Eq.~(\ref{eq:robust-coverage}).
\end{proof}

\begin{remark}[Reading $\delta$ empirically: a lower bound, not an estimate]
  \label{rem:empirical-delta}
  Theorem~\ref{thm:robust-coverage} does not certify $\delta$ for an unknown
  attacker. Because Eq.~(\ref{eq:robust-coverage}) rearranges to
  $\delta \geq (1-\alpha) - \Pr(y' \in \mathcal{C}_\alpha(x'))$, an observed
  coverage gap gives only a \emph{lower bound} on $\delta$ for the evaluated
  attack, never an upper bound or a point estimate, and never a certificate
  against a stronger, unevaluated adversary: a large $\delta$ merely fails
  to \emph{guarantee} high coverage, it does not force coverage to be low.
  Consequently, an attack that leaves ta-CP coverage close to nominal yields
  an uninformative (near-zero) bound, while only a substantial coverage drop
  yields a nontrivial one. The label-free XGBoost surrogate-transfer attacks
  (supplemental material) leave coverage close to nominal even at worst-of-20
  budget and $66$--$72\%$ attack success, so their certified $\delta$ is
  near-zero and uninformative; the direct, target-informed attack
  (Table~\ref{tab:direct-attack}), which queries the target LLM's own score
  at every round, drives coverage substantially below nominal and yields the
  only nontrivial bound in this paper: $\delta \geq 0.075$ for LLaMA3-8B and
  $\delta \geq 0.19$ for Qwen3-8B. This asymmetry indicates that direct query
  access to the target's own score, not candidate budget or attack success
  rate alone, is the threat-model dimension that matters most here.
\end{remark}

\subsection{Exact Coverage via Attack-Orbit-Invariant Representations}
\label{sec:method:invariance}

Theorem~\ref{thm:robust-coverage} bounds the coverage gap in terms of a
total variation distance $\delta$ that can only be certified as a lower
bound after the attack is observed (Remark~\ref{rem:empirical-delta}). This
section shows that the DC/IC/UC taxonomy already introduced in
Section~\ref{sec:threatmodel:taxonomy} supports a strictly stronger
guarantee for a restricted representation of the input: not a bound on the
gap, but its elimination, together with the precise condition under which
this is possible and a single identity that recovers
Theorem~\ref{thm:robust-coverage} as a special case.

\begin{definition}[Attack orbit]
  \label{def:orbit}
  Under the oblivious DC adversary (Definition~\ref{def:adversary}), the
  \emph{attack orbit} of a flow $x$ is
  \begin{equation}
    O(x) = \{x\} \cup \{\mathcal{P}_{\mathcal D}(x+\delta) :
    \delta \in \Delta_\mathrm{DC}(\varepsilon)\}.
    \label{eq:orbit}
  \end{equation}
\end{definition}

\begin{definition}[Orbit-invariant representation]
  \label{def:invariant-rep}
  A measurable map $\phi : \mathcal{X} \to \mathcal{Z}$ is
  \emph{orbit-invariant} if $\phi(x') = \phi(x)$ for every $x \in \mathcal{X}$
  and every $x' \in O(x)$.
\end{definition}

By construction, $\mathcal{P}_{\mathcal D}$ can change a coordinate only if it
is a DC coordinate or a deterministic function of DC coordinates through an
audited recompute relation (e.g., a rate recomputed as count over duration,
where count is DC); every other coordinate is fixed pointwise by
$\mathcal{P}_{\mathcal D}$ for every $\delta \in \Delta_\mathrm{DC}(\varepsilon)$.
This does \emph{not} mean every non-DC feature is orbit-invariant in the
sense of Definition~\ref{def:invariant-rep}: a feature outside the DC mask
that is nonetheless a deterministic descendant of a DC feature (e.g., the
recomputed rate itself) still changes over the orbit. We therefore refine
the DC/IC/UC taxonomy operationally into three empirically determinable
roles for a given feature set: DC features, deterministic descendants of DC
features, and the residual, orbit-invariant coordinates. Because the
recompute relations enforced by $\mathcal{P}_{\mathcal D}$ are implemented
procedurally rather than published as closed-form dependency graphs, we
identify descendants by probing $\mathcal{P}_{\mathcal D}$ directly
(Algorithm~\ref{alg:descendant}) rather than assuming a known dependency
graph.

\begin{algorithm}[t]
  \caption{Empirical Descendant Detection}
  \label{alg:descendant}
  \begin{algorithmic}[1]
    \Require Held-out probe rows $\{x_j\}_{j=1}^m$; projection
      $\mathcal{P}_{\mathcal D}$; DC mask $\mathbf{m}_\mathrm{DC}$;
      perturbation budget $\varepsilon$; number of trials $T$.
    \Ensure Descendant coordinate set $\mathcal{R}$.
    \State $\mathcal{R} \gets \emptyset$
    \For{$t = 1$ to $T$}
      \For{$j = 1$ to $m$}
        \State Draw $\delta_j \sim \mathrm{Uniform}(-\varepsilon\,\mathbf{m}_\mathrm{DC},\,
               +\varepsilon\,\mathbf{m}_\mathrm{DC})$
        \State $x_j' \gets \mathcal{P}_{\mathcal D}(x_j+\delta_j)$
        \For{each coordinate $k \notin \mathbf{m}_\mathrm{DC}$}
          \If{$x_j'[k] \neq x_j[k]$}
            \State $\mathcal{R} \gets \mathcal{R} \cup \{k\}$
          \EndIf
        \EndFor
      \EndFor
    \EndFor
    \State \Return $\mathcal{R}$ \Comment{descendants; UC-only set is
      $\{1,\ldots,d\} \setminus (\mathbf{m}_\mathrm{DC} \cup \mathcal{R})$}
  \end{algorithmic}
\end{algorithm}

Algorithm~\ref{alg:descendant} empirically constructs a candidate
orbit-invariant coordinate set: it flags a coordinate as a descendant, and
removes it, if it changes under \emph{any} of the $T$ independent random DC
perturbations tried across the probe set, so every dependency the probe
observes triggers removal. Finite probing cannot certify the converse---a
descendant relation realized only outside the $T$ sampled perturbations, or
under a probe-set input the $T$ trials never covered, would not be flagged,
so the returned coordinate set is not guaranteed to be a subset of the true
orbit-invariant coordinates in general. Theorem~\ref{thm:invariant-coverage}
gives an exact guarantee only for a $\phi$ that is genuinely orbit-invariant;
Section~\ref{sec:exp:invariance} verifies that the coordinate set
Algorithm~\ref{alg:descendant} returns for RT-IoT2022 remains byte-identical
under the specific threat configuration evaluated there ($T=8$ probe trials,
$\varepsilon=0.15$, and the worst-of-20 attack), not that it is certified
invariant against every conceivable perturbation the DC/IC/UC threat model
permits. $\phi_{\mathrm{UC}}$ denotes the projection of $x$ onto this
UC-only coordinate set.

\begin{theorem}[Exact Coverage via Orbit-Invariant Representation]
  \label{thm:invariant-coverage}
  Let $\phi$ be orbit-invariant (Definition~\ref{def:invariant-rep}), and let
  split CP (Section~\ref{sec:background:cp}) be applied with a nonconformity
  score $s(\phi(x),y)$ that depends on $x$ only through $\phi(x)$. Let
  $(x_i,y_i)_{i=1}^{n+1}$ be exchangeable, let calibration scores be computed
  from $\phi(x_i)$ for $i \leq n$ (on clean or matched-perturbed calibration
  data), and let $\hat{q}$ be the resulting threshold. Then, for \emph{any}
  adversary $a$ with $a(x) \in O(x)$---randomized, adaptive, white-box, and
  possibly aware of the calibration set, $\hat{q}$, and the true test
  label---the following holds for every realization, not merely in
  distribution:
  \begin{equation}
    \mathcal{C}_{\hat{q}}(a(x_{n+1})) = \mathcal{C}_{\hat{q}}(x_{n+1}),
    \label{eq:pathwise-identity}
  \end{equation}
  and consequently
  \begin{equation}
    \Pr\!\left(y_{n+1} \in \mathcal{C}_{\hat{q}}(a(x_{n+1}))\right)
    = \Pr\!\left(y_{n+1} \in \mathcal{C}_{\hat{q}}(x_{n+1})\right)
    \geq 1-\alpha,
    \label{eq:invariant-coverage}
  \end{equation}
  with \emph{no} gap term, regardless of the adversary's power.
\end{theorem}

\begin{proof}
  Since $s$ depends on $x$ only through $\phi(x)$, orbit-invariance gives
  $s(\phi(x'),y) = s(\phi(x),y)$ for every $y$ and every
  $x' \in O(x)$, hence for any fixed threshold $t$,
  $\mathcal{C}_t(x') = \{y : s(\phi(x'),y)\leq t\} =
  \{y: s(\phi(x),y)\leq t\} = \mathcal{C}_t(x)$: the prediction-set map is
  constant on orbits for every threshold. If calibration is computed on
  clean data, $\hat{q}$ is a function of $(\phi(x_i),y_i)_{i\leq n}$ only and
  is therefore independent of the test-time attack; applying orbit-constancy
  at $t=\hat{q}$, $x = x_{n+1}$, $x'=a(x_{n+1})\in O(x_{n+1})$ gives
  Eq.~\eqref{eq:pathwise-identity} pointwise, and $y_{n+1}$ is untouched by
  the attack, so the coverage indicator is identical outcome by outcome;
  taking expectations and invoking the standard split-CP guarantee on the
  clean side gives Eq.~\eqref{eq:invariant-coverage}. If calibration is
  itself matched-perturbed as in Algorithm~\ref{alg:tacp}, each calibration
  score $s(\phi(\mathcal{P}_{\mathcal D}(x_i+\xi_i)),y_i)$ equals
  $s(\phi(x_i),y_i)$ by the same orbit-invariance, so the multiset of
  calibration scores, and hence $\hat{q}$, is unchanged; the argument reduces
  to the clean-calibration case. No measurability or independence
  restriction on $a$ is required for Eq.~\eqref{eq:pathwise-identity}, since
  it is a pointwise identity for every fixed outcome; measurability of $a$ is
  only needed for the probability in Eq.~\eqref{eq:invariant-coverage} to be
  well defined, and there the event coincides with the clean event, which is
  measurable.
\end{proof}

Theorem~\ref{thm:invariant-coverage} is qualitatively different from
Theorem~\ref{thm:robust-coverage}: the conclusion is a pathwise identity, not
a probabilistic bound, so no amount of adversary knowledge (of the
calibration set, $\hat{q}$, or the true label) can subvert it, because there
is no longer anything probabilistic left for the adversary to exploit. This
is also the reason the theorem fails without orbit-invariance: if $\phi$ is
not orbit-invariant, a white-box adversary can set
$a(x_{n+1}) = \arg\max_{x'\in O(x_{n+1})} s(\phi(x'),y_{n+1})$, which is
exactly the mechanism Theorem~\ref{thm:robust-coverage} can only bound and
that Table~\ref{tab:direct-attack} and
Table~\ref{tab:invariant-representations} (Full-feature rows) realize
empirically as coverage collapse.

\begin{theorem}[Gap Identity]
  \label{thm:gap-identity}
  Let $\hat{q}$ be fixed given the calibration sample and let
  $S=s(\phi(X),Y)$ be the clean true-label score. For a point $(x,y)$ define
  the \emph{orbit score-inflation}
  \begin{equation}
    G(x,y) = \sup_{x'\in O(x)} s(\phi(x'),y) - s(\phi(x),y) \;\geq\; 0,
    \label{eq:orbit-inflation}
  \end{equation}
  which is identically zero whenever $\phi$ is orbit-invariant (the converse
  need not hold: $G\equiv 0$ only requires the induced score, not $\phi$
  itself, to be constant on orbits---see Remark~\ref{rem:invariant-scope}).
  Against the worst-case adaptive adversary, the coverage lost relative to
  the clean distribution satisfies the \emph{exact} identity
  \begin{equation}
    \mathrm{Gap} := \Pr(S \leq \hat{q}) - \Pr(S+G(X,Y) \leq \hat{q})
    = \Pr\!\left(\hat{q}-G(X,Y) < S \leq \hat{q}\right).
    \label{eq:gap-identity}
  \end{equation}
  If, in addition, the density of $S$ is bounded by $L$ in a neighborhood of
  $\hat{q}$ and $G \leq \bar{G}$ almost surely, then
  $\mathrm{Gap} \leq L\bar{G}$; this bound is $0$ exactly when $\bar{G}=0$,
  though $\mathrm{Gap}$ itself can be $0$ for a specific distribution even
  when $\bar{G}>0$.
\end{theorem}

\begin{proof}
  The worst-case attacked coverage is
  $\Pr(\sup_{x'\in O(X)} s(\phi(x'),Y) \leq \hat{q}) =
  \Pr(S+G(X,Y)\leq\hat{q})$ by definition of $G$. Subtracting from the clean
  coverage $\Pr(S\leq\hat{q})$, and using $\{S+G\leq\hat q\}\subseteq
  \{S\leq\hat q\}$ since $G\geq 0$, gives
  $\mathrm{Gap} = \Pr(S\leq\hat q) - \Pr(S+G\leq\hat q)
  = \Pr(\hat q - G < S \leq \hat q)$, the stated identity. Conditioning on
  $G\leq\bar G$, $\Pr(\hat q-G<S\leq\hat q) \leq \Pr(\hat q - \bar G < S \leq
  \hat q) = F_S(\hat q) - F_S(\hat q-\bar G) \leq L\bar G$ under the bounded-density
  assumption; $\bar G=0$ empties the interval and gives $\mathrm{Gap}=0$.
\end{proof}

\begin{remark}[Relation to Theorem~\ref{thm:robust-coverage} and existing
  robust conformal prediction]
  \label{rem:gap-unifies}
  Eq.~\eqref{eq:gap-identity} is a single functional with two regimes.
  Existing robust conformal prediction methods
  \cite{gendler2022adversarially,ghosh2023probabilistically,yan2024provably,
  zargarbashi2025robust,massena2025efficient,luo2026gametheoretic} fix the
  score's input representation and \emph{bound} $\bar{G}$ (via a Lipschitz
  constant, a smoothing radius, or a quantile-of-quantile construction), then
  read off $\mathrm{Gap}\leq L\bar{G}$: the $\bar{G}>0$ regime, with a
  necessarily probabilistic certificate. Theorem~\ref{thm:invariant-coverage}
  is the achievability of the complementary $\bar{G}=0$ endpoint of the same
  functional, obtained by choosing $\phi$ so the adversary's action is the
  identity map on it rather than by bounding the score's sensitivity to it.
  Theorem~\ref{thm:robust-coverage} is the special case in which $\bar G$ is
  left unbounded and the shift is measured directly in total variation
  distance rather than through $G$.
\end{remark}

\begin{remark}[Scope of the invariant-representation guarantee: sufficiency,
  not an unconditional converse]
  \label{rem:invariant-scope}
  Orbit-invariance of $\phi$ is \emph{sufficient} for exact zero-gap coverage
  (Theorem~\ref{thm:invariant-coverage}), but it is not necessary: a
  representation can vary over an orbit and still yield zero gap, provided
  the induced \emph{decision} $\mathcal{C}_{\hat q}$ does not change, i.e.,
  the variation stays within a level set of the score. What \emph{is} true
  is a completeness statement restricted to the class of orbit-invariant
  procedures: writing $\sim$ for the equivalence relation generated by
  $x' \in O(x)$ and $\pi$ for the associated quotient map, every
  orbit-invariant $\phi$ factors through $\pi$, so $I(\phi(X);Y) \leq
  I(\pi(X);Y)$; $\pi$ is therefore the information-maximal orbit-invariant
  representation, and no orbit-invariant procedure can be more efficient
  than split CP applied to $\pi(X)$ (or a sufficient statistic of it, such as
  $\phi_\mathrm{UC}$). This statement is scoped to the orbit-invariant class;
  it does not assert that every zero-gap procedure is orbit-invariant, and we
  do not claim a converse beyond this scope.
\end{remark}

The remainder of the paper evaluates $\phi_\mathrm{UC}$, the empirical
UC-only representation constructed by Algorithm~\ref{alg:descendant}, as a
concrete instance of Theorem~\ref{thm:invariant-coverage} on RT-IoT2022
(Section~\ref{sec:exp:invariance}), together with a representation that
removes only the literal DC coordinates without descendant detection (which
Remark~\ref{rem:invariant-scope} predicts need not be orbit-invariant, and is
evaluated to quantify how large a residual gap Eq.~\eqref{eq:gap-identity}
permits when it is not).

\subsection{Abstention as an Uncertainty Signal}
\label{sec:method:abstention}

When $|\mathcal{C}(x)|=0$, no class meets the calibrated confidence threshold.
We treat this, and any $|\mathcal{C}(x)|\neq1$ outcome, as an explicit
escalation trigger to a human analyst rather than a null prediction, turning
the CP set into a principled operational rule without any additional
hyperparameter. The supplemental material summarizes which claims in this
paper are formally guaranteed versus empirical, and reports the escalation
rate observed under each experiment.

\section{Experimental Evaluation}
\label{sec:experiments}

\subsection{Setup}
\label{sec:exp:setup}

\textbf{Datasets.}
We evaluate on three public network intrusion benchmarks:
\textbf{CIC-IDS-2018}~\cite{sharafaldin2018cicids} (25 flow-level features;
five labels, but the held-out test partition contains no Probe samples, so
CIC coverage numbers are four-class estimates with Probe retained only as a
candidate label); \textbf{RT-IoT2022} (83 FlowMeter features from IoT/IIoT
traffic; all five mapped classes evaluated); and \textbf{HIKARI-2021}~\cite{hikari2021}
(\texttt{HIKARI2022.csv}, 228{,}253 rows, 88 source columns; three observed
mapped classes, with all five labels remaining candidates at inference).

\textbf{Models.}
The main analysis uses LLaMA3-8B~\cite{llama3} checkpoints trained separately
for each dataset. Fixing the main analysis to one architecture keeps the
score definition, prompt boundary, split construction, and checkpoint
provenance fixed while testing whether the proposed calibration mechanism
behaves consistently across datasets. Section~\ref{sec:exp:crossarch} reports
cross-architecture replications with Qwen3-8B~\cite{qwen3} and
Mistral-7B~\cite{mistral7b} checkpoints, each trained under the identical
protocol, partition, and sampling seeds.

\textbf{CP Protocol.}
We employ split CP (Algorithm~\ref{alg:tacp}) with $\alpha = 0.05$ throughout,
targeting 95\% marginal coverage. Fine-tuning and held-out evaluation are
separated by the exact serialized prompt under a fixed partition seed, so no
model-visible prompt group occurs in both pools; prompt groups with
conflicting labels are excluded before sampling. Calibration and test
perturbations are independent uniform draws from the same DC-feature
distribution at $\varepsilon\in\{0.15,0.30\}$ in normalized coordinates,
inverse-transformed and projected by $\mathcal{P}_{\mathcal D}$, which
audits bounds, domains, and dataset-specific relations; these runs test the
matched-distribution claim in Theorem~\ref{thm:ta-coverage} and are not
optimized attacks. The main runs employ $n_\mathrm{cal}=n_\mathrm{test}=300$
and report mean $\pm$ std across sampling seeds 42, 123, and 456. Full
provenance auditing (row and prompt hashing, excluded-group counts,
checkpoint metadata) is detailed in the supplemental material.

\textbf{Metrics.}
We report marginal coverage, average prediction set size, singleton/empty-set/escalation
rates (escalation meaning $|\mathcal{C}(x)|\ne 1$), and the quantile
threshold $\hat q$, all computed on the test split.

\subsection{Calibration Mismatch and Coverage}
\label{sec:exp:tacpresults}

Table~\ref{tab:tacp} asks whether a threshold calibrated on clean traffic covers
scores after the matched DC perturbation and whether recalibrating on the same
perturbation law restores marginal coverage. Standard CP loses coverage on all
three benchmarks. Ta-CP moves coverage to approximately the nominal 0.95 level,
consistent with Theorem~\ref{thm:ta-coverage}.

\begin{table*}[t]
  \centering
  \caption{Matched-perturbation results for LLaMA3-8B ($\alpha=0.05$,
           $n_\mathrm{cal}=n_\mathrm{test}=300$). Values are mean $\pm$ standard
           deviation across three sampling seeds. ``Std. cov.'' uses the clean
           threshold on perturbed test rows. ``ta cov.'' and ``ta size'' employ a
           threshold calibrated on independent perturbations from the same law.}
  \label{tab:tacp}
  \small
  \setlength{\tabcolsep}{4pt}
  \begin{tabular}{lccccc}
    \toprule
    Dataset & $\varepsilon$ & Std. cov. & ta cov. & ta size & Escalation \\
    \midrule
    HIKARI-2021 & 0.15 & $0.8489\pm0.0185$ & $0.9456\pm0.0087$ & $1.0444\pm0.0253$ & $0.0444\pm0.0253$ \\
                & 0.30 & $0.8367\pm0.0109$ & $0.9467\pm0.0082$ & $1.0722\pm0.0134$ & $0.0722\pm0.0134$ \\
    \midrule
    RT-IoT2022 & 0.15 & $0.0900\pm0.0152$ & $0.9667\pm0.0094$ & $4.2656\pm0.0397$ & $0.9644\pm0.0057$ \\
               & 0.30 & $0.0800\pm0.0196$ & $0.9422\pm0.0273$ & $4.2489\pm0.1476$ & $0.9422\pm0.0368$ \\
    \midrule
    CIC-IDS-2018 & 0.15 & $0.3900\pm0.0216$ & $0.9500\pm0.0170$ & $4.2122\pm0.3535$ & $1.0000\pm0.0000$ \\
                & 0.30 & $0.3611\pm0.0211$ & $0.9678\pm0.0150$ & $4.5411\pm0.1534$ & $1.0000\pm0.0000$ \\
    \bottomrule
  \end{tabular}
\end{table*}

The coverage result alone is insufficient. HIKARI-2021 recovers coverage with
near-singleton sets, whereas RT-IoT2022 and CIC-IDS-2018 recover coverage mainly
by returning multiple labels. The RT and CIC results are coverage recovery but
not efficient decision recovery.

\subsection{Ranking Loss Explains the Efficiency Boundary}

To distinguish threshold mismatch from loss of class information, we rank all
five candidate labels by their NLL and record the true-label rank. Table~\ref{tab:rank}
shows that HIKARI true labels remain within the top two under every evaluated
perturbation. RT and CIC true labels frequently fall below the top two.
A conformal threshold can include these labels, but it cannot restore the
classifier's ranking; large valid sets are therefore expected.

\begin{table}[t]
  \centering
  \caption{True-label ranking under matched perturbations. Each row pools three
           sampling seeds (900 test rows).}
  \label{tab:rank}
  \setlength{\tabcolsep}{4pt}
  \begin{tabular}{lccccc}
    \toprule
    Dataset & $\varepsilon$ & Top-1 & Top-2 & Top-3 & Top-4 \\
    \midrule
    HIKARI & 0.15 & 0.9233 & 1.0000 & 1.0000 & 1.0000 \\
           & 0.30 & 0.9056 & 1.0000 & 1.0000 & 1.0000 \\
    RT-IoT & 0.15 & 0.1478 & 0.5211 & 0.7744 & 0.9189 \\
           & 0.30 & 0.1433 & 0.4633 & 0.7456 & 0.9178 \\
    CIC-IDS & 0.15 & 0.3211 & 0.5900 & 0.8744 & 0.9222 \\
            & 0.30 & 0.3022 & 0.5900 & 0.8744 & 0.9222 \\
    \bottomrule
  \end{tabular}
\end{table}

This separation yields the main empirical boundary: traffic-aware calibration
repairs coverage loss caused by a matched score shift, but efficient recovery
requires the perturbed model to preserve the true label near the top of its
candidate ranking. HIKARI is the efficient regime in the evaluated setting;
RT and CIC are coverage-only regimes.

\begin{figure}[t]
  \centering
  \includegraphics[width=\columnwidth]{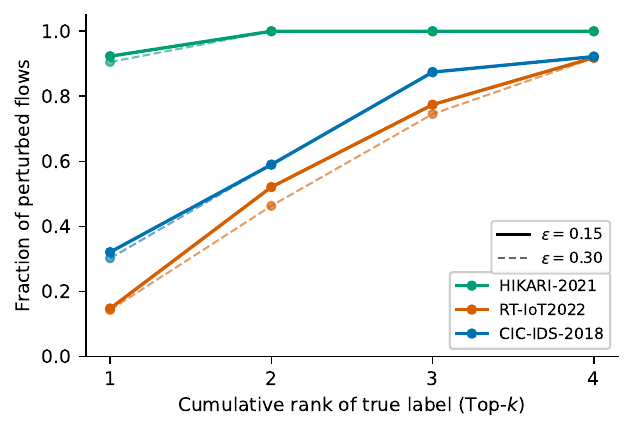}
  \caption{Cumulative true-label rank under matched DC perturbation
           (Table~\ref{tab:rank}), pooled across three sampling seeds
           (900 test rows per curve). HIKARI true labels remain within the
           top two candidates almost always; RT-IoT2022 and CIC-IDS-2018 true
           labels frequently fall to rank three or four, which is why ta-CP
           must return large sets to cover them.}
  \label{fig:rank-distribution}
\end{figure}

\subsection{Classwise Coverage, Mondrian Calibration, and a Non-LLM Baseline}
\label{sec:exp:classwise}

The marginal-only guarantee of Theorem~\ref{thm:ta-coverage} does not certify
classwise coverage: under marginal ta-CP, specific rare classes (HIKARI-2021
Exploitation, CIC-IDS-2018 CredentialAccess) fall as low as $0.48$--$0.60$
coverage even though marginal coverage is near nominal, and increasing the
calibration sample size does not resolve this. We prove a class-conditional
(Mondrian) extension of Theorem~\ref{thm:ta-coverage} that computes a
separate perturbed-calibration threshold per label and certifies per-class
coverage under the same matched-perturbation assumption, handling the
finite-sample edge case in which a class has too few calibration examples
for the required order statistic to exist. Applying it to the same
calibration and test rows already collected raises the worst observed
per-class coverage from $0.48$ to $0.98$ (HIKARI-2021) and from $0.60$ to
$0.94$ (CIC-IDS-2018), at a real efficiency cost on HIKARI-2021 and close to
no additional cost on the two benchmarks where marginal ta-CP was already
least efficient. The full theorem statement and proof, the per-class
breakdown, and the corresponding table and figure are given in the
supplemental material.

To situate the LLM results against a non-LLM classifier, we train an
XGBoost classifier on the identical partition, group-disjoint sampling, and
perturbation protocol, substituting its predicted class probabilities into
the same nonconformity score. XGBoost reproduces the identical qualitative
pattern observed for both LLM architectures on all three benchmarks:
standard CP collapses under matched DC perturbation, ta-CP restores coverage
to approximately nominal, and the same coverage-through-large-sets boundary
separates HIKARI-2021 from RT-IoT2022 and CIC-IDS-2018. Because XGBoost and
the two LLM architectures are unrelated model families with unrelated score
constructions and all reproduce the same per-dataset coverage-efficiency
split under the identical protocol, the boundary is best explained as a
property of the DC-perturbed traffic/label relationship in these benchmarks,
not an artifact of one classifier family. The full results table is given in
the supplemental material.

\subsection{Comparison to Weighted Conformal Prediction}
\label{sec:exp:weightedcp}

The calibration mismatch in Section~\ref{sec:threatmodel:degradation} is an
instance of covariate shift, for which an existing, off-the-shelf remedy
already exists: weighted split CP~\cite{tibshirani2019conformal} reweights
calibration nonconformity scores by an estimated density ratio $w(x) =
dQ/dP(x)$ between the test-time covariate distribution $Q$ and the
calibration-time distribution $P$, without modifying the calibration data
itself. This raises a direct question for the proposed method: does
reweighting already solve the coverage problem that ta-CP solves by
recalibrating on perturbed traffic?

We evaluate weighted CP under the identical matched-perturbation scenario as
Table~\ref{tab:tacp}: clean calibration, DC-perturbed test, same partition,
sampling seeds, and $\varepsilon$ values, so that all three methods---standard
CP, weighted CP, and ta-CP---are compared on exactly the same calibration and
test rows. The density ratio is estimated with the standard classifier trick:
a logistic regression is trained to discriminate clean-calibration covariates
from perturbed-test covariates using only the DC-feature coordinates (the
only coordinates the perturbation touches), giving $\hat w(x) = \hat
h(x)/(1-\hat h(x))$ from the classifier's predicted probability $\hat h(x)$;
the weighted conformal quantile follows~\cite[Eq.~3]{tibshirani2019conformal}.
No LLM inference is required for this comparison: the already-computed,
promoted LLaMA3-8B nonconformity scores from Table~\ref{tab:tacp} are reused
directly, and only the calibration weighting changes.

\begin{table}[t]
  \centering
  \caption{Weighted CP~\cite{tibshirani2019conformal} under the identical
           matched-perturbation scenario as Table~\ref{tab:tacp}, LLaMA3-8B
           ($\alpha=0.05$, mean $\pm$ std across the same three sampling
           seeds). ``Std.\ cov.'' and ``ta cov.'' are repeated from
           Table~\ref{tab:tacp} for direct comparison.}
  \label{tab:weightedcp}
  \small
  \setlength{\tabcolsep}{3pt}
  \begin{tabular}{lcccc}
    \toprule
    Dataset & $\varepsilon$ & Std.\ cov. & Weighted cov. & ta cov. \\
    \midrule
    HIKARI-2021  & 0.15 & 0.8489 & $0.8956\pm0.0057$ & 0.9456 \\
                 & 0.30 & 0.8367 & $0.8967\pm0.0094$ & 0.9467 \\
    \midrule
    CIC-IDS-2018 & 0.15 & 0.3900 & $0.3856\pm0.0244$ & 0.9500 \\
                 & 0.30 & 0.3611 & $0.3533\pm0.0260$ & 0.9678 \\
    \midrule
    RT-IoT2022   & 0.15 & 0.0900 & $0.5056\pm0.0706$ & 0.9667 \\
                 & 0.30 & 0.0800 & $0.5089\pm0.2496$ & 0.9422 \\
    \bottomrule
  \end{tabular}
\end{table}

Table~\ref{tab:weightedcp} shows that weighted CP does not close the coverage
gap that ta-CP closes, and its behavior differs sharply by dataset. On
HIKARI-2021, weighted CP improves over standard CP ($0.849$--$0.837$ to
$0.896$--$0.897$) but still falls short of both the $1-\alpha=0.95$ target
and ta-CP. On CIC-IDS-2018, weighted CP provides \emph{no improvement at
all} over standard CP ($0.386$--$0.353$ versus $0.390$--$0.361$): the two are
statistically indistinguishable given the observed seed variation. On
RT-IoT2022, weighted CP recovers substantially over standard CP's near-total
collapse ($0.09$--$0.08$ to $0.51$--$0.51$) but remains far below nominal and,
critically, becomes highly \emph{unstable}: at $\varepsilon=0.30$ the three
seeds give coverage $0.16$, $0.73$, and $0.64$ (std $=0.250$, an order of
magnitude larger than any other coverage estimate reported in this paper). We
attribute this pattern to a known failure mode of weighted CP under severe
covariate shift: when clean-calibration and DC-perturbed-test covariates
become close to linearly separable in the DC-feature subspace (a positivity,
or overlap, violation), the estimated density ratio concentrates almost all
calibration weight mass on a small, seed-dependent subset of calibration
points, so the effective calibration sample size collapses and the resulting
weighted quantile is both biased and high-variance. CIC-IDS-2018 has the
highest DC-feature dimensionality of the three benchmarks and shows the
smallest improvement over standard CP, consistent with this explanation.
Ta-CP does not face this failure mode because it does not reweight existing
calibration data by an estimated, potentially near-degenerate density ratio;
it generates new calibration data directly from the perturbation law itself,
so calibration and (matched) test scores are exactly exchangeable by
construction rather than approximately so via an estimated correction.

\subsection{Cross-Architecture Replication and Surrogate-Transfer Attacks}
\label{sec:exp:crossarch}

A matched-protocol study restricted to one architecture cannot distinguish
the proposed mechanism from an architecture-specific artifact. Qwen3-8B and
Mistral-7B checkpoints, each trained under the identical corrected protocol
and evaluated with the same sampling seeds, epsilons, and calibration and
test sizes, reproduce the same qualitative boundary observed for LLaMA3-8B:
standard CP collapses on all three benchmarks, ta-CP restores marginal
coverage to approximately nominal, HIKARI-2021 recovers coverage with
near-singleton sets, and RT-IoT2022 and CIC-IDS-2018 recover coverage only
through broad sets. The near-identical class-level undercoverage pattern
across three independently fine-tuned architectures from three different
model families indicates the boundary is a property of the perturbed-traffic/label
relationship in these benchmarks, not an artifact of one checkpoint. Full
per-architecture tables are given in the supplemental material.

We additionally test whether the HIKARI-2021 result persists when candidate
perturbations are selected by a separate classifier rather than sampled
once, using an XGBoost surrogate trained only on the training partition to
select, from independently projected candidates, the one that most reduces
its own clean-prediction probability. A weak instance of this attack
(worst-of-1, changing about 5\% of clean-correct predictions on HIKARI-2021)
and a substantially stronger instance (worst-of-20, changing 66--72\% of
clean-correct predictions on CIC-IDS-2018, the benchmark where the
coverage-efficiency boundary is most severe) both leave ta-CP coverage close
to nominal in both architectures, despite an order-of-magnitude difference
in attack strength; by Remark~\ref{rem:empirical-delta}, both therefore
certify only a near-zero, uninformative lower bound on $\delta$. This
motivates evaluating a stronger, direct attack next. Full results are given
in the supplemental material.

\subsection{A Direct, Target-Informed Attack}
\label{sec:exp:direct-attack}

Every attack above is surrogate-transfer: candidate selection never queries
the target LLM's own scores, only an XGBoost proxy's clean-prediction
probability. This leaves open whether an attacker with direct query access to
the target model's output scores---but not its parameters, gradients, or the
calibrated threshold $\hat q_\mathrm{ta}$---could degrade ta-CP coverage
further. We term this a \emph{score-based target-query} attack, distinct
from a gradient-based white-box attack or a threshold-aware attack (neither
of which we evaluate). At each of $k=20$ rounds, a candidate DC perturbation
is scored by the target LLM itself, and the candidate that maximizes the
target's own true-label NLL is kept, using the true label but no surrogate
model at any stage---a standard, realistic assumption for characterizing
worst-case degradation, though stronger than a real field attacker's access.
Calibration uses the same single random i.i.d.\ DC draw as the
matched-perturbation experiment (Table~\ref{tab:tacp}), so only the
test-time attack differs. We run this on CIC-IDS-2018 (the dataset where the
coverage-efficiency boundary is most severe) for both LLaMA3-8B and
Qwen3-8B, across the same three sampling seeds, epsilons, and provenance
gates used throughout.

\begin{table}[t]
  \centering
  \caption{CIC-IDS-2018 direct, target-informed worst-of-20 attack: the
           target LLM's own true-label NLL is queried at every round, with no
           surrogate model ($\alpha=0.05$, mean $\pm$ std across three sampling
           seeds). Attack success is the fraction of clean-correct target
           predictions made incorrect.}
  \label{tab:direct-attack}
  \scriptsize
  \setlength{\tabcolsep}{2pt}
  \begin{tabular}{lccccc}
    \toprule
    Model & $\varepsilon$ & Attack success & Std.\ cov. & ta cov. & ta size \\
    \midrule
    LLaMA3-8B & 0.15 & $0.826\pm0.016$ & $0.224\pm0.011$ & $0.876\pm0.011$ & $3.980\pm0.237$ \\
              & 0.30 & $0.820\pm0.025$ & $0.212\pm0.016$ & $0.874\pm0.011$ & $4.168\pm0.141$ \\
    Qwen3-8B  & 0.15 & $0.826\pm0.013$ & $0.286\pm0.010$ & $0.788\pm0.032$ & $3.142\pm0.167$ \\
              & 0.30 & $0.800\pm0.019$ & $0.297\pm0.021$ & $0.760\pm0.009$ & $3.263\pm0.015$ \\
    \bottomrule
  \end{tabular}
\end{table}

The direct attack is both stronger and more damaging than either
surrogate-transfer instance. Attack success rises to $78$--$83\%$ of
clean-correct predictions (versus $66$--$72\%$ for worst-of-20
surrogate-transfer on the same dataset), and, critically, ta-CP coverage no
longer stays close to nominal: it falls to $0.874$--$0.876$ for LLaMA3-8B and
$0.760$--$0.788$ for Qwen3-8B, a drop of $6$--$19$ percentage points below the
$1-\alpha=0.95$ target---the only attack evaluated here under which ta-CP
coverage degrades by more than a few percentage points. This shows direct
query access to the target's own score is meaningfully stronger than
transferring from a surrogate, even at matched candidate budget, and that
the two architectures are not equally robust to it: Qwen3-8B's coverage drop
($16$--$19$ points) is roughly double LLaMA3-8B's ($6$--$8$ points) at
matched $\varepsilon$, despite both showing the same qualitative
coverage-efficiency boundary under matched random perturbations. By
Remark~\ref{rem:empirical-delta}, this certifies the only nontrivial lower
bound obtained in this paper: $\delta \geq 0.075$--$0.076$ for LLaMA3-8B and
$\delta \geq 0.16$--$0.19$ for Qwen3-8B.

\subsection{Attack-Orbit-Invariant Representations: An Exact Empirical Guarantee}
\label{sec:exp:invariance}

Section~\ref{sec:method:invariance} shows that scoring on an orbit-invariant
representation $\phi$ eliminates the coverage gap exactly
(Theorem~\ref{thm:invariant-coverage}), rather than merely bounding it
(Theorem~\ref{thm:robust-coverage}), and that a representation which removes
only literal DC coordinates without descendant detection is not guaranteed
to be orbit-invariant (Remark~\ref{rem:invariant-scope}). We evaluate both
claims on RT-IoT2022 under the strongest attack in this paper, the direct,
target-informed, adaptive target-score-query worst-of-20 attack of
Section~\ref{sec:exp:direct-attack} (true label used, target score queried,
no parameters or gradients), applied here to each representation's own
prompt in turn rather than only the standard one.

\textbf{Representations compared.} We fine-tune and evaluate three prompt
variants on RT-IoT2022, holding the partition, sampling seeds, LoRA
configuration, and training schedule identical to the main experiments
(Section~\ref{sec:exp:setup}): \textbf{Full}, the standard eight-line,
twenty-field verbalization used throughout the paper (identical to the
checkpoints behind Table~\ref{tab:tacp}); \textbf{DC-naive}, the same
template with every literal DC field replaced by a constant but every
other field---including aggregate rate and flag fields that are
deterministic descendants of DC fields under $\mathcal{P}_{\mathcal
D}$---left untouched; and \textbf{UC-only} ($\phi_\mathrm{UC}$), the
residual field set returned by Algorithm~\ref{alg:descendant} after probing
$\mathcal{P}_{\mathcal D}$ with $8$ independent random DC perturbations per
probe row (flow duration, backward packet count, backward payload total,
backward PSH-flag count, and the initial backward TCP window size), which
excludes all forward-direction fields and every field found to change under
probing.
Each variant is fine-tuned independently (LLaMA3-8B for all three variants;
Qwen3-8B and Mistral-7B for Full and UC-only) with $1{,}000$ examples per
class ($5$ epochs), under the same group-disjoint sampling used for the main
checkpoints. Evaluation uses $n_\mathrm{cal}=n_\mathrm{test}=300$,
$\alpha=0.05$, three sampling seeds, and the $K=20$ direct worst-of-$k$
attack of Section~\ref{sec:exp:direct-attack}, applied here to each
representation's own prompt (so, e.g., the UC-only attack searches for a DC
perturbation that changes the UC-only prompt string, not the standard one).
For the coarser DC-naive and UC-only prompts, calibration and test rows are
drawn only from held-out-pool prompt groups that never occur anywhere in the
training pool under that representation's own prompt text, which is a
strictly stronger disjointness guarantee than the row-level check used
elsewhere in the paper and is necessary because a coarser representation can
map multiple distinct raw rows to identical prompt strings.

\begin{table}[t]
  \centering
  \caption{Attack-orbit invariant representations on RT-IoT2022 under the
           $K=20$ direct worst-of-20 attack ($\varepsilon=0.15$, $\alpha=0.05$,
           mean $\pm$ std across three sampling seeds, $n_\mathrm{cal}=
           n_\mathrm{test}=300$). ``Identical'' reports whether every one of
           the $20$ attack rounds produced a prompt string identical to the
           clean prompt, verified exactly rather than estimated.}
  \label{tab:invariant-representations}
  \scriptsize
  \setlength{\tabcolsep}{2pt}
  \begin{tabular}{llcccccc}
    \toprule
    Repr.\ & Model & Clean top-1 & Clean cov.\ & Atk.\ top-1 & Atk.\ cov.\ & Atk.\ size & Identical \\
    \midrule
    Full & LLaMA3-8B & $0.994\pm0.004$ & $0.943\pm0.035$ & $0.120\pm0.017$ & $0.006\pm0.004$ & $0.840\pm0.024$ & No \\
    Full & Qwen3-8B  & $0.994\pm0.002$ & $0.964\pm0.015$ & $0.062\pm0.012$ & $0.001\pm0.002$ & $0.796\pm0.026$ & No \\
    Full & Mistral-7B& $0.996\pm0.003$ & $0.954\pm0.023$ & $0.053\pm0.008$ & $0.000\pm0.000$ & $0.463\pm0.139$ & No \\
    \midrule
    DC-naive & LLaMA3-8B & $0.942\pm0.006$ & $0.943\pm0.021$ & $0.932\pm0.006$ & $0.938\pm0.022$ & $1.040\pm0.059$ & No \\
    \midrule
    UC-only & LLaMA3-8B & $0.857\pm0.035$ & $0.947\pm0.019$ & $0.857\pm0.035$ & $0.947\pm0.019$ & $1.436\pm0.059$ & \textbf{Yes} \\
    UC-only & Qwen3-8B  & $0.913\pm0.021$ & $0.952\pm0.006$ & $0.913\pm0.021$ & $0.952\pm0.006$ & $1.136\pm0.055$ & \textbf{Yes} \\
    UC-only & Mistral-7B& $0.922\pm0.015$ & $0.952\pm0.016$ & $0.922\pm0.015$ & $0.952\pm0.016$ & $1.099\pm0.038$ & \textbf{Yes} \\
    \bottomrule
  \end{tabular}
\end{table}

\textbf{The Full representation collapses.} Consistent with
Table~\ref{tab:direct-attack} and Theorem~\ref{thm:robust-coverage}, scoring
on the unrestricted feature set gives near-perfect clean accuracy ($0.994$--
$0.996$) but collapses under attack across all three architectures: attacked
top-1 accuracy falls to $0.05$--$0.12$ and attacked coverage falls to
$0.000$--$0.006$, a $94$--$96$ percentage-point coverage loss from a clean
level already near nominal. This holds for every architecture tested, so the
collapse is a property of scoring on the unrestricted representation under
this attack, not of one checkpoint.

\textbf{The UC-only representation is exactly invariant, empirically, not
just in expectation.} Across all $3$ architectures $\times$ $3$ seeds $= 9$
independent runs and $20$ attack rounds each ($180$ attack attempts total),
every UC-only attacked prompt was byte-identical to its clean counterpart:
the ``Identical'' column in Table~\ref{tab:invariant-representations} is Yes
in every case, and attacked top-1 accuracy, coverage, and average set size
match their clean counterparts to machine precision in all $9$ runs. This is
the empirical realization of the pathwise identity in
Eq.~\eqref{eq:pathwise-identity}: the $K=20$ direct attack cannot move the
UC-only representation off its clean value by construction, so there is
nothing for the worst-case search to find. The utility cost of this exact
guarantee is visible in the Clean top-1 column: UC-only clean accuracy
($0.857$--$0.922$) is $7$--$14$ percentage points below Full's clean accuracy
($0.994$--$0.996$), and clean prediction sets are correspondingly larger
($1.10$--$1.44$ vs.\ $0.94$--$0.96$ under Full, before Full's coverage
collapses under attack). Theorem~\ref{thm:invariant-coverage}'s guarantee is
therefore not free: it trades some of the base model's ranking information
for an unconditional guarantee, precisely the trade the taxonomy in
Section~\ref{sec:threatmodel:taxonomy} makes explicit.

\textbf{The DC-naive representation is not exactly invariant, and the gap it
leaves is small but consistently nonzero.} Removing only the literal DC
fields, without descendant detection, does \emph{not} reproduce the Full
collapse: attacked top-1 accuracy ($0.932\pm0.006$) and attacked coverage
($0.938\pm0.022$) remain close to their clean values
($0.942\pm0.006$ and $0.943\pm0.021$). However, the ``Identical'' column is
No in every seed: the attacked prompt does change under the $K=20$ search
(through the retained descendant fields), and Table~\ref{tab:invariant-representations}
shows this produces a small, consistently nonzero coverage loss of $0.5$--$1.0$
percentage points, not the $94$--$96$-point collapse of the Full
representation. This is direct evidence for Remark~\ref{rem:invariant-scope}:
$\phi$-invariance is sufficient but not automatic from removing DC fields
alone, and the residual descendant leakage that Algorithm~\ref{alg:descendant}
is designed to catch is exactly what separates the DC-naive row from the
UC-only rows in Table~\ref{tab:invariant-representations}.

A point-level verification of Eq.~\eqref{eq:gap-identity} for the DC-naive
representation---confirming that the observed coverage gap equals exactly
the number of test points whose true-label score crosses $\hat q$ under
attack, divided by $n_\mathrm{test}$, in every seed---is given in the
supplemental material.

\begin{remark}[Scope of Section~\ref{sec:exp:invariance}]
  \label{rem:invariance-scope}
  This evaluation is restricted to RT-IoT2022 and a single architecture for
  the DC-naive comparison, and Algorithm~\ref{alg:descendant}'s guarantee is
  verified against the specific probe budget and attack evaluated here, not
  certified against every conceivable one; Section~\ref{sec:discussion:limits}
  details these boundaries.
\end{remark}

\section{Discussion}
\label{sec:discussion}

\subsection{Deployment Considerations}

Coverage is not decision recovery. Matching the calibration distribution to
the test perturbation law repairs the threshold mismatch and marginal
coverage; the usefulness of the resulting set still depends on the base
model's candidate ranking, which calibration alone cannot restore.
HIKARI-2021 preserves the true label near the top of the ranking and yields
near-singleton sets, while RT-IoT2022 and CIC-IDS-2018 often rank it lower,
so ta-CP can only cover the label by widening the set---a boundary on what
calibration can recover, not a failure of the conformal guarantee.

Whether a high escalation rate signals an attack or merely reflects an
inflated perturbed-calibration threshold depends on the (dataset,
architecture) pair (supplemental material): a deployer should verify
clean-traffic escalation before treating a high rate as attack evidence.
Ta-CP and standard CP share the same test-time scoring workload, and
calibration is an offline procedure, making the mechanism compatible with
batch-mode traffic analysis; recalibration frequency must be set from
observed drift rather than a fixed interval, since the guarantee requires
exchangeability after the specified perturbation mechanism. Three backbones
in the 7--8B range support a cross-dataset mechanism claim that is not
specific to one checkpoint, though Section~\ref{sec:exp:direct-attack} shows
architecture choice still affects empirical robustness to a strong attack.

\subsection{Limitations}
\label{sec:discussion:limits}

Theorem~\ref{thm:ta-coverage} and Theorem~\ref{thm:robust-coverage} concern
an adversary captured by an i.i.d.\ or distributionally bounded perturbation
process; the surrogate-transfer and direct attacks evaluated here optimize
candidates adaptively and are therefore outside Theorem~\ref{thm:ta-coverage}'s
scope, though within Theorem~\ref{thm:robust-coverage}'s as evidence about
$\delta$. A threshold-aware adversary that additionally targets $\hat
q_\mathrm{ta}$ directly could plausibly induce a larger $\delta$ than any
attack evaluated here, and is not tested. The numerical perturbation budget
bounds the DC proposal before semantic projection, so dependent derived
coordinates can move by more than $\varepsilon$ once relational consistency
is enforced; the experiments therefore characterize the explicit
proposal-plus-projection law, not a final-vector norm ball. CIC-IDS-2018
yields no evaluated Probe-class flows, and NLL-based scores remain sensitive
to candidate token length despite a fixed label vocabulary.

The attack-orbit-invariant evaluation (Section~\ref{sec:exp:invariance}) is
restricted to RT-IoT2022; the DC-naive comparison uses a single
architecture, and whether its small, nonzero gap generalizes to the other
two is not established. Algorithm~\ref{alg:descendant} is empirical and
probe-based, so a dependency realized only outside the probed perturbation
magnitudes would not be caught; Table~\ref{tab:invariant-representations}'s
guarantee is verified against the specific attack evaluated, not certified
against every conceivable probe budget. Most importantly,
Theorem~\ref{thm:invariant-coverage}'s exact guarantee is a property of the
chosen representation, not of the deployed model as a whole: it does not
transfer to a deployment that scores the full representation for accuracy
and only \emph{post hoc} restricts to the invariant one for some inputs.
The supplemental material details these boundaries further.

\section{Conclusion}
\label{sec:conclusion}

Traffic-aware conformal prediction addresses a specific failure of
conformal intrusion detection: a threshold estimated on clean traffic does
not cover nonconformity scores once the traffic distribution shifts under
controllable-feature perturbation. When calibration and test perturbations
are i.i.d.\ draws from the same distribution, ta-CP inherits split CP's
finite-sample marginal coverage guarantee; a distributionally robust
extension further bounds the coverage loss by the total variation distance
between calibration- and test-time score distributions for any test-time
process, realized empirically in both directions: surrogate-transfer
attacks leave coverage close to nominal and certify only an uninformative
near-zero bound, while a direct attack querying the target model's own
score is markedly stronger and yields the only nontrivial certified bound
in this paper. Weighted conformal prediction, evaluated under the identical
scenario, does not close this gap and becomes unstable under severe
covariate shift, so the mechanism is not a relabeling of an existing remedy.

This distributionally robust bound still must be certified after the fact
from an observed coverage drop. We show it can be eliminated rather than
merely tightened for representations invariant to the adversary's action:
scoring on such a representation yields an exact, pathwise coverage identity
against any adversary within the threat model, and a gap identity unifies
this result with the distributionally robust bound and with existing
Lipschitz- and smoothing-based robust conformal prediction methods as the
two endpoints of one functional---bound the adversary's reach, or eliminate
it by construction. An automated, probe-based construction of such a
representation on RT-IoT2022 leaves every evaluated attack attempt across
three architectures exactly unchanged, while a naively restricted
representation that omits descendant detection leaves a small but
consistently nonzero gap---evidence that descendant detection, not merely
removal of attacker-controlled fields, is what the exact guarantee
requires. This guarantee costs seven to fourteen points of clean accuracy
relative to the unrestricted feature set, a trade a deployer must weigh
explicitly, and its scope is currently limited to one benchmark and to the
coordinate set the descendant-detection procedure returns under its
evaluated trial budget.

\section*{Data Availability Statement}
CIC-IDS-2018, HIKARI-2021, and RT-IoT2022 are publicly available benchmark
datasets. Feature annotations, certification code, and per-sample
predictions will be released upon acceptance.

\bibliographystyle{IEEEtran}
\bibliography{refs}

\end{document}